\documentclass{article} 

\usepackage{amsmath}
\usepackage{amssymb}
\usepackage{dsfont}
\usepackage{bbm}
\usepackage{tikz}

\usepackage[margin=3cm]{geometry}
\usepackage[round, sort&compress]{natbib}
\usepackage{har2nat}
\usepackage{amsthm}
\newtheorem{theorem}{Theorem}
\newtheorem{lemma}{Lemma}
\newtheorem{corollary}{Corollary}
\newtheorem{proposition}{Proposition}
\theoremstyle{definition}
\newtheorem{definition}{Definition}
\newtheorem{remark}{Remark}

\usepackage{bm}

\usepackage[plain]{algorithm2e} 

\DeclareMathOperator{\Prob}{\mathbb{P}}
\DeclareMathOperator{\E}{\mathbb{E}}
\DeclareMathOperator{\Et}{\mathbb{E}_t}
\newcommand{\flnw}{\lfloor N w_t^{(i)} \rfloor}
\newcommand{\1}[1]{\mathbbm{1}_{\{#1\}}}

\begin{document}




\title{Simple conditions for convergence of sequential Monte Carlo genealogies with applications}

\author{S. Brown$^{1}$\\s.brown.18@warwick.ac.uk \and P. A. Jenkins$^{1,2,3}$\\p.jenkins@warwick.ac.uk \and A. M. Johansen$^{1,3}$\\a.m.johansen@warwick.ac.uk \and J. Koskela$^{1}$\\j.koskela@warwick.ac.uk} 
{\date{\small $^1$ Department of Statistics, University of Warwick, Coventry, U.K.\\[2pt]
$^2$ Department of Computer Science, University of Warwick, Coventry, U.K.\\[2pt]
$^3$ Alan Turing Institute, London, U.K.}}

\maketitle

\begin{abstract}
We present simple conditions under which the limiting genealogical process associated with a class of interacting particle systems with \emph{non-neutral} selection mechanisms, as the number of particles grows, is a time-rescaled Kingman coalescent. Sequential Monte Carlo algorithms are popular methods for approximating integrals in problems such as non-linear filtering and smoothing which employ this type of particle system. Their performance depends strongly on the properties of the induced genealogical process. We verify the conditions of our main result for standard sequential Monte Carlo algorithms with a broad class of low-variance resampling schemes, as well as for conditional sequential Monte Carlo with multinomial resampling.
\end{abstract}


\section{Introduction}
Interacting particle systems (IPSs) are used as models in many situations in which one wants to study the macroscopic effects of a system defined via a large number of microscopic interactions. Such models have been used to explain phenomena as diverse as tumour growth, behavioural systems, and epidemics. See \cite{lig:2005} for an extensive treatment of such systems in continuous time and \cite{delmoral2004} for a class of discrete time systems of interest in this article.
As particles reproduce and die out, there naturally arises an embedded genealogical history of the system, and this genealogy is of interest both in its own right and as a tool for understanding the evolution of the system. This is particularly true in population genetics, where a broad class of neutral IPS models for an evolving natural population converge, in the infinite population limit and after a suitable time-rescaling, to a well-studied random genealogical process, the so-called \emph{$n$-coalescent} first described by \cite{kingman1982coal}. One might ask: in what other contexts does the $n$-coalescent, or another genealogical process, arise as the scaling limit of a particular IPS? Our goal in this paper is to answer this question within the context of \emph{sequential Monte Carlo} (SMC), a broad class of stochastic algorithms used for computational statistical inference. These algorithms are best known for their application in non-linear filtering and smoothing \cite{gordon1993}, but can be applied generally to mean-field approximation of Feynman-Kac flows; see \cite{delmoral2004} for more background. These methods have found diverse applications throughout signal processing, statistics, econometrics, biology, and many other disciplines, and understanding their properties is of widespread importance; see, for example, \cite{chopin2020,doucet2009,fearnhead18,naesseth2019} for recent surveys.

In SMC, a population of particles evolves in discrete time. The algorithm proceeds by iterating through two steps: a mutation step in which the positions of the particles are updated by applying some Markov kernel; and a selection step in which the particles are weighted by some potential function and resampled to form the next generation. Resampling stochastically duplicates high-weight particles and removes low-weight particles. The duplicates of a particle in the following generation are termed its offspring. A succession of resampling steps induces a genealogy, that is, a process recording the parent-offspring relationships between particles at consecutive generations.

SMC genealogies are important because they capture the phenomenon of \emph{ancestral degeneracy}, which has a substantial impact on the performance of the algorithm.  Due to resampling, the number of distinct ancestors whose descendants comprise the particles at the terminal time decays rapidly as the time horizon increases, a well known consequence of which is that path-based Monte Carlo estimators typically have high variance (see, for example, \cite{briers10,fearnhead18}).
The theoretical analysis of SMC genealogies is a useful tool for addressing a number of practical questions in SMC methods. It provides a framework for answering questions such as: for a fixed number of particles, what is the maximum time horizon over which smoothing estimators are reliable? This is important for tuning the parameters in settings such as particle MCMC \cite{andrieu2010}, fixed-lag smoothing, and stable variance estimation \cite{olsson2019}. It also provides estimates of the memory cost associated with path storage, building on the results of \cite{jacob2015} and \cite[Corollary 2]{koskela2018}.

In this paper we provide simple sufficient conditions under which the genealogy of an SMC algorithm converges to a time-rescaled $n$-coalescent \cite{kingman1982coal} in the large population limit. We require control over only the second and third moments of the marginal family size of each parent. 
This builds upon \cite{koskela2020annals}, a slight error in which was corrected in \cite{koskela2018}.
Ours is a substantial improvement over that work, which requires additional control over fourth moments, including cross-terms, to obtain the same convergence result. That work also imposes a condition on the speed of convergence that is violated for instance by the neutral Moran model, any finite sample of which is known to have a Kingman genealogy in the large population limit (see \cite[p47]{durrett2008}). Our result covers algorithms that the existing results do not and also relies on simple conditions which admit verification for a broad class of algorithms.

The result is known to apply to standard SMC with multinomial resampling \cite[Corollary 1]{koskela2018}. 
We additionally prove convergence for any resampling scheme based on stochastic rounding, described in Definition \ref{defn:stochround}. This includes low-variance schemes such as systematic resampling, residual resampling with stratified residuals, and the more exotic schemes proposed in \cite{crisan1997} and \cite{gerber2017}. The results presented in this paper therefore provide the first a priori characterization of the genealogy of an SMC algorithm for resampling schemes that are widely used in practice, offering quantitative insight into the nature of ancestral degeneracy and into the relative performance of different resampling schemes. We also show that our result applies to \emph{conditional} SMC, a variant that forms a building block of the particle Gibbs algorithm \cite{andrieu2010}, in which it is important to maintain at least two distinct ancestors across a fixed time window. That our results apply to conditional SMC is important because for many practical statistical problems comprising a high-dimensional latent state space model, particle Gibbs is one of the few algorithms that is able to explore parameter posteriors in reasonable time, so understanding the factors affecting its mixing time is of great interest.

\section{Sequential Monte Carlo}
Algorithm \ref{alg:SMC} describes sequential Monte Carlo with $N$ particles over a fixed time window $T$.
The initial proposal distribution is $\mu$, $(K_t)_{t=1}^T$ is a sequence of Markov transition kernels with respective transition densities $(q_t)_{t=1}^T$, and $(g_t)_{t=0}^T$ is a sequence of potential functions.
For simplicity these may be assumed to exist on a common state space that is a subspace of $\mathbb{R}^d$, so $g_0: \mathbb{R}^d \to \mathbb{R}_+$ and, for $t>0$, $g_t: \mathbb{R}^d \times\mathbb{R}^d \to \mathbb{R}_+$ and $q_t: \mathbb{R}^d \times\mathbb{R}^d \to \mathbb{R}_+$,
although the state spaces can in general be any sequence of Polish spaces. We allow potential functions at time $t \geq 1$ to depend upon both the current and previous value of the state to allow our results to directly cover several classes of algorithm which occur in practice, including particle filters in which a proposal other than the transition density are used (see, for example, \cite[Section 4.1]{doucet2009}, \cite[Section 10.3.2]{chopin2020}) and sequential Monte Carlo samplers \cite{delmoral2006}.
At generation $t$, $w_t^{(1:N)} = (w_t^{(1)},\ldots,w_t^{(N)})$ is the vector of particle weights, and $a_t^{(1:N)} = (a_t^{(1)},\ldots,a_t^{(N)})$ is the vector of resampled parental indices.

The output of Algorithm~\ref{alg:SMC} will depend on the particular application, but might include the complete collection of sampled state values, $\{X_t^{(i)}: t\in\{0,\ldots,T\}, i \in\{1,\ldots,N\}\}$ and parental indices $\{a_t^{(i)}:t\in\{0,\ldots,T-1\},N\in\{1,\ldots,N\}\}$. These algorithms are widely used in particular for inference in general state space hidden Markov models where the principal estimation problems are \textit{filtering}, using
$\sum_{i=1}^N w_t^{(i)} \delta_{X_t^{(i)}}$ at time $t$; \textit{smoothing} using
  $\sum_{i=1}^N w_t^{(i)} \delta_{\tilde{X}_{t,0:t}^{(i)}}$
  where the path-particles $\tilde{X}_{t:0:t}^{(i)}$ are vectors in $\mathbb{R}^{d\times(t+1)}$ obtained by setting $\tilde{X}_{t,t}^{(i)} = X_{t}^{(i)}$ and setting $X_{t,s}^{(i)}$ for each $s < t$ to the state associated with the time $s$ ancestor of particle $i$ at time $t$;
  and computation of the \emph{marginal likelihood} via $$\hat{Z}_T := \frac{1}{N} \sum_{i=1}^N g_0(X_0^{(i)}) \prod_{t=1}^T \frac{1}{N} \sum_{i=1}^N g_t(X_{t-1}^{(a_{t-1}^{(i)})}, X_t^{(i)}).$$ Of these, smoothing most explicitly depends upon samples from the full path $X_{0:T}$, but information about the paths may also be retained in the other settings in order to estimate variance \cite{lee2018, olsson2019} and the associated genealogies influence the sampling properties of all such estimators.

The variety of possible procedures for the \textsc{resample} step is an active area of research. Some important examples are explored in Section \ref{sec:applications}.
We  take valid resampling schemes to be those where: the total number of resampled offspring is $N$; the expected number of offspring of particle $i$ conditional on the weights is $N w_t^{(i)}$; and each offspring is assigned equal weight after resampling.

\vspace{10pt}
\begin{algorithm}
\DontPrintSemicolon
\KwIn{$N, T, \mu, (K_t)_{t=1}^T, (g_t)_{t=0}^T$}
\lFor{$i \in \{1,\dots,N\}$}{ 
	Sample $X_0^{(i)} \sim \mu(\cdot)$
}
\lFor{$i \in \{1,\dots, N\}$}{
		$w_{0}^{(i)} \gets  \left\{{\sum_{j=1}^N g_0(X_0^{(j)})}\right\}^{-1}{g_0(X_0^{(i)})} $ 
	}
\For{$t \in \{0,\dots, T-1\}$}{
	Sample $a_t^{(1:N)} \sim $ \textsc{resample}$(\{1,\dots ,N\}, w_t^{(1:N)}$)\;
	\lFor{$i \in \{1,\dots,N\}$}{
		Sample $X_{t+1}^{(i)} \sim K_{t+1}(X_t^{(a_t^{(i)})}, \cdot)$
	}
	\lFor{$i \in \{1,\dots, N\}$}{	
		$w_{t+1}^{(i)} \gets \Big\{ {\sum_{j=1}^Ng_{t+1}(X_t^{(a_t^{(j)})},X_{t+1}^{(j)}) }\Big\}^{-1} g_{t+1}(X_t^{(a_t^{(i)})},X_{t+1}^{(i)}) $
	}
}
\caption{Sequential Monte Carlo}\label{alg:SMC}
\end{algorithm}
\vspace{10pt}

Under the standing assumption stated below, it is sufficient for our purposes to consider the vector $\nu_t^{(1:N)} = (\nu_t^{(1)},\dots,\nu_t^{(N)})$ of offspring counts, where $\nu_t^{(i)} = |\{ j\in \{1,\dots,N\} : a_t^{(j)} =i \}|$.

\textbf{Standing Assumption:} The conditional distribution of parental indices $a_t^{(1:N)}$ given offspring counts $\nu_t^{(1:N)}$ is uniform over all valid assignments.

The standing assumption is a weaker condition than exchangeability \cite[p446]{mohle1998}.
Any resampling scheme can be made to satisfy it by applying an additional permutation of the offspring indices after selecting the parents (see \cite[p. 290]{andrieu2010}).

\section{A limit theorem for sequential Monte Carlo genealogies}
For convenience, we will henceforth measure time backwards, with the terminal particles at time 0 and the initial particles at time $T$. The forward-time process of particles replicating or dying induces a coalescent process when viewed in backwards in time.
The full forward-time process is Markovian, but this no longer holds after integrating out the positions of particles and their weights.
Thus, the reverse-time process of ancestral lineages without position or weight information is not Markovian either.

We will analyse an asymptotic regime in which the total number of particles $N\to\infty$, and consider the finite-dimensional restriction to a sample of $n\leq N$ terminal particles.
The genealogy of such a sample is described by a partition-valued stochastic process $(G_t^{(n,N)})_{t=0}^T$. At time 0 its value is the set of singletons $\{\{1\},\dots,\{n\}\}$. At each time $t$, $i\neq j \in \{1,\dots, n\}$ belong to the same partition block  in $G_t^{(n,N)}$ if and only if terminal particles $i$ and $j$ share a common ancestor at time $t$. We will take $N\to\infty$ and show that the $n$-coalescent \cite{kingman1982coal} is the correct limiting object.
Our limit theorem will apply after a rescaling of time in which the genealogy is viewed over an infinite time horizon; that is, $T\to\infty$.

\begin{definition}
\label{def:kingman}
The \emph{$n$-coalescent} is the homogeneous continuous-time Markov process on the set of partitions of $\{1,\dots,n\}$ with infinitesimal generator $Q$ having entries
\begin{equation*}
(Q)_{\xi,\eta} = \begin{cases}
1 & \xi \prec \eta\\
-|\xi|(|\xi|-1)/2 & \xi=\eta \\
0 & \text{otherwise}
\end{cases}
\end{equation*}
where $\xi$ and $\eta$ are partitions of $\{1,...,n\}$, $|\xi|$ denotes the number of blocks in $\xi$, and $\xi \prec \eta$ means that $\eta$ is obtained from $\xi$ by merging exactly one pair of blocks.
\end{definition}

Throughout the following, falling factorials are denoted by $(a)_b = a(a-1) \cdots (a-b+1)$. We denote by $\mathcal{F}_{t} = \sigma(\nu_s^{(1:N)} : 1 \leq s \leq t)$ the filtration generated by offspring counts, and use the shorthand $\Et[\cdot] \equiv \E[ \cdot \mid \mathcal{F}_{t-1} ]$ for filtered expectations. Since time is labelled in reverse, the filtrations contain information about the future of the original system rather than the past.

A central quantity for analysing convergence of these coalescent processes is the \emph{pair merger rate}, which can be interpreted as the probability that a randomly chosen pair of particles at time $t-1$ have a common ancestor at time $t$. Conditional on $\mathcal{F}_t$, this probability is
\begin{equation}\label{eq:coalescence_rate}
c_N(t) := \frac{1}{(N)_2}\sum_{i=1}^N (\nu_t^{(i)})_2 .
\end{equation}
In the $n$-coalescent the pair merger rate is equal to 1. 
Thus, as $N\to\infty$, the time scaling required to possibly obtain a Kingman coalescent limit is the generalized inverse
\begin{equation*}
\tau_N(t) := \inf \left\{ s\geq 1 : \sum_{r=1}^s c_N(r) \geq t \right\} .
\end{equation*}
Following \cite{mohle1998}, we exclude the case where $\Prob[ \tau_N(t) = \infty ] >0$ for finite $t$. 
This is a very mild assumption which in reality is violated only by pathological (from this perspective) models, for example where minimum-variance resampling is used and the potentials are constant. Settings in which this assumption are violated are unlikely to be encountered in any practical application of SMC. For the classes of algorithms given as examples in Section~\ref{sec:applications} we provide tractable sufficient conditions under which the assumption holds, and these are explicitly verified in Appendix~\ref{app:finiteness}.

In the $n$-coalescent, there are almost surely no mergers involving more than two lineages at a time. 
As shown in \cite[Lemma 1, Case 3]{koskela2018}, an upper bound on the conditional probability that more than two lineages merge at time $t$ in the pre-limiting model is
\begin{equation*}
D_N(t) := \frac{1}{N(N)_2} \sum_{i=1}^N (\nu_t^{(i)})_2 \left\{\nu_t^{(i)} + \frac{1}{N} \sum_{j\neq i} (\nu_t^{(j)})^2 \right\} .
\end{equation*}
This includes both the possibility of three or more lineages merging into one, and of two or more simultaneous pair mergers.
Theorem \ref{thm:generalconv}, our main result, gives a simple sufficient condition controlling these merger rates to yield a Kingman limit for particle genealogies.

\begin{theorem}\label{thm:generalconv}
Let $\nu_t^{(1:N)}$ denote the offspring numbers in an IPS satisfying the standing assumption and such that, for any $N$ sufficiently large, $\Prob[ \tau_N(t) = \infty ] =0$ for all finite $t$. Suppose that there exists a deterministic sequence $(b_N)_{N\geq1}$ such that ${\lim}_{N\to\infty} b_N =0$ and
\begin{equation}\label{eq:mainthmcond}
\frac{1}{(N)_3} \sum_{i = 1}^N \Et\left[ (\nu_t^{(i)})_3 \right]  \leq b_N \frac{1}{(N)_2} \sum_{i = 1}^N \Et\left[ (\nu_t^{(i)})_2 \right]
\end{equation}
for all $N$, uniformly in $t \geq 1$.
Then the rescaled genealogical process $(G_{\tau_N(t)}^{(n,N)})_{t\geq0}$ converges in the sense of finite-dimensional distributions to Kingman's $n$-coalescent as $N \to \infty$.
\end{theorem}

\begin{proof}[Proof outline]
Here we present the high level argument. In the interests of clarity, details are deferred to a sequence of lemmata which follow this result, and to Section~\ref{sec:proof}.

Theorem \ref{thm:generalconv} has the same conclusion as \cite[Theorem 1]{koskela2018} for which the conditions are
\begin{align}
\E[ c_N(t) ] &\rightarrow 0 , \label{eq:oldass1}\\
\E \left[ \sum_{ r = \tau_N( s ) + 1 }^{ \tau_N( t ) } D_N( r ) \right] &\rightarrow 0, \label{eq:oldass2}\\
\E \left[ \sum_{ r = \tau_N( s ) + 1 }^{ \tau_N( t ) } c_N( r )^2 \right] &\rightarrow 0 , \label{eq:oldass3}\\
\E [ \tau_N(t) - \tau_N(s) ] &\leq C_{t,s} N; \label{eq:oldass4}
\end{align}
as $N\to\infty$, for some strictly positive constant $C_{t,s}$ that does not depend on $N$. We show that the more tractable conditions of the present theorem are sufficient for the result.

The series of Lemmata \ref{lem:removeass3}--\ref{lem:removeass2} below show that the assumptions \eqref{eq:oldass1}--\eqref{eq:oldass3} follow from \eqref{eq:mainthmcond}. Lemma \ref{lem:removeass4} allows us to remove condition \eqref{eq:oldass4} by improving upon some arguments from the proof of \cite[Theorem 1]{koskela2018}; this argument is presented in detail in Section~\ref{sec:proof}.
\end{proof}

Our result improves on \cite{koskela2018} by eliminating the restrictive condition \eqref{eq:oldass4}, which is shown in Lemma~\ref{lem:removeass4} to be unnecessary. This allows our result to apply to some models not previously included; for example the neutral Moran model of population genetics violates \eqref{eq:oldass4} but is covered by Theorem~\ref{thm:generalconv}. 
In neutral models the straightforward analogue of \eqref{eq:mainthmcond} is both necessary and sufficient \cite[Theorem 5.4]{mohle2003}, suggesting that in general this condition is not significantly stronger than \eqref{eq:oldass1}--\eqref{eq:oldass3} combined.

\begin{lemma} \label{lem:removeass3}
$\eqref{eq:oldass2} \Rightarrow \eqref{eq:oldass3}$.
\end{lemma}

\begin{proof}
It is sufficient to show that $c_N( t )^2 \leq D_N( t ) N/(N-1)$.
We have
\begin{align*}
c_N( t )^2 &= \frac{ 1 }{ N ( N - 1 ) ( N )_2 } \sum_{ i = 1 }^N ( \nu_t^{(i)})_2 \Bigg\{ \nu_t^{(i)} ( \nu_t^{(i)} - 1 ) + \sum_{\substack{j=1\\ j \neq i }}^N ( \nu_t^{(j)} )_2 \Bigg\} \\
&= \frac{ 1 }{ N ( N )_2 } \sum_{ i = 1 }^N ( \nu_t^{(i)} )_2 \Bigg\{ \frac{ \nu_t^{(i)} ( \nu_t^{(i)} - 1 ) }{ N - 1 } + \frac{ 1 }{ N - 1 } \sum_{\substack{j=1\\ j \neq i }}^N ( \nu_t^{(j)} )_2 \Bigg\} \\
&\leq \frac{ 1 }{ N ( N )_2 } \sum_{ i = 1 }^N ( \nu_t^{(i)})_2 \Bigg\{ \nu_t^{(i)} + \frac{ 1 }{ N - 1 } \sum_{\substack{j=1\\ j \neq i }}^N ( \nu_t^{(j)} )_2 \Bigg\} \\
&\leq \frac{ 1 }{ N ( N )_2 } \sum_{ i = 1 }^N ( \nu_t^{(i)})_2 \Bigg\{ \nu_t^{(i)} + \frac{ N / ( N - 1 ) }{ N } \sum_{\substack{j=1\\ j \neq i }}^N ( \nu_t^{(j)} )^2 \Bigg\} \\
&\leq \frac{ N / ( N - 1 ) }{ N ( N )_2 } \sum_{ i = 1 }^N ( \nu_t^{(i)})_2 \Bigg\{ \nu_t^{(i)} + \frac{ 1 }{ N } \sum_{\substack{j=1\\ j \neq i }}^N ( \nu_t^{(j)} )^2 \Bigg\} = \frac{ N }{ N - 1 } D_N( t )
\end{align*}
which concludes the proof.
\end{proof}

\begin{lemma} \label{lem:removeass1}
$\eqref{eq:mainthmcond} \Rightarrow \eqref{eq:oldass1}$.
\end{lemma}

\begin{proof}
Following the proof of \cite[Lemma 5.5]{mohle2003}, we fix $\epsilon > 0$ and define the event $A_i = \{ \nu_t^{(i)} \leq N \epsilon \}$.
Now
\begin{align}
\Et[ c_N( t ) ] &= \frac{ 1 }{ ( N )_2 } \sum_{ i = 1 }^N \Et\left[ ( \nu_t^{(i)} )_2 \right] 
= \frac{ 1 }{ ( N )_2 } \sum_{ i = 1 }^N \left\{ \Et\left[ ( \nu_t^{(i)} )_2 \mathds{ 1 }_{ A_i } \right] + \Et\left[ ( \nu_t^{(i)} )_2 \mathds{ 1 }_{ A_i^c } \right] \right\} \nonumber \\
&\leq \frac{ \epsilon }{ N - 1 } \sum_{ i = 1 }^N \Et\left[ \nu_t^{(i)} \mathds{ 1 }_{ A_i } \right] + \sum_{ i = 1 }^N \Et\left[ \mathds{ 1 }_{ A_i^c } \right] \nonumber \\
&\leq \{ 1 + O( N^{ -1 } ) \} \epsilon + \sum_{ i = 1 }^N \Prob\left[ \nu_t^{(i)} > N \epsilon \mid \mathcal{F}_{t-1} \right]. \label{cond_cN}
\end{align}
For $N \geq 3 / \epsilon$, Markov's inequality yields
\begin{align}
\sum_{ i = 1 }^N \Prob\left[ \nu_t^{(i)} > N \epsilon \mid \mathcal{F}_{t-1} \right] &\leq \frac{ 1 }{ ( N \epsilon )_3 } \sum_{ i = 1 }^N \Et\left[ ( \nu_t^{(i)} )_3\right] = \frac{ \{ 1 + O( N^{ -1 } ) \} }{ \epsilon^3 ( N )_3 } \sum_{ i = 1 }^N \Et\left[ ( \nu_t^{(i)} )_3 \right] \nonumber \\
&\leq \{ 1 + O( N^{ -1 } ) \} \frac{ b_N }{ \epsilon^3 } \Et[ c_N( t ) ]. \label{markovs_ineq}
\end{align}
Substituting \eqref{markovs_ineq} into \eqref{cond_cN} and using $c_N( t ) \leq 1$ results in
\begin{equation*}
\Et[ c_N( t ) ] \leq \{ 1 + O( N^{ -1 } ) \} \Bigg( \epsilon + \frac{ b_N }{ \epsilon^3 } \Bigg) \underset{N\to\infty}{\longrightarrow} \epsilon
\end{equation*}
since $b_N \rightarrow 0$. 
As $\epsilon > 0$ was arbitrary, we have
$
\E[ c_N( t ) ] = \E\left[ \Et[ c_N( t ) ] \right ] \rightarrow 0
$
as $N \rightarrow \infty$.
\end{proof}

\begin{lemma} \label{lem:removeass2}
$\eqref{eq:mainthmcond} \Rightarrow \eqref{eq:oldass2}$.
\end{lemma}

\begin{proof}
We decompose $D_N(t)$ as the sum of two terms and consider their filtered expectations. The first is
\begin{align}
\frac{ 1 }{ N ( N )_2 } \sum_{ i = 1 }^N \Et\left[ ( \nu_t^{(i)} )_2 \nu_t^{(i)}\right] &= \frac{ 1 }{ N ( N )_2 } \sum_{ i = 1 }^N \Et\left[ 2 ( \nu_t^{(i)} )_2 + ( \nu_t^{(i)} )_3 \right] \nonumber \\
&\leq \frac{ 2 }{ N } \Et[ c_N( t ) ] + \frac{ 1 }{ ( N )_3 } \sum_{ i = 1 }^N \Et\left[ ( \nu_t^{(i)} )_3 \right] \nonumber \\
&\leq \Bigg(\frac{ 2 }{ N } + b_N \Bigg) \Et[ c_N( t ) ]. \label{DN_part_1}
\end{align}
The second is
\begin{align}
\frac{ 1 }{ N^2 ( N )_2 } \sum_{ j=1 }^N \sum_{ i \neq j } \Et\left[ ( \nu_t^{(i)} )_2 ( \nu_t^{(j)} )^2 \right] &= \frac{ 1 }{ N^2 ( N )_2 } \sum_{ j=1 }^N \sum_{ i \neq j } \Et\left[ ( \nu_t^{(i)} )_2 ( \nu_t^{(j)} )_2 + ( \nu_t^{(i)} )_2 \nu_t^{(j)} \right] \nonumber \\
&\leq \frac{ 1 }{ N^2 ( N )_2 } \sum_{ j=1 }^N \sum_{ i \neq j } \Et\left[ ( \nu_t^{(i)} )_2 ( \nu_t^{(j)} )_2 \right] + \frac{ \Et[ c_N( t ) ] }{ N }. \label{DN_part_2}
\end{align}
Now, with $A_i$ defined as in Lemma \ref{lem:removeass1},
\begin{align}
\sum_{ j=1 }^N \sum_{ i \neq j } \Et\left[ ( \nu_t^{(i)} )_2 ( \nu_t^{(j)} )_2\right] &= \sum_{ j=1 }^N \sum_{ i \neq j } \left\{ \Et\left[ ( \nu_t^{(i)} )_2 ( \nu_t^{(j)} )_2 \mathds{ 1 }_{ A_i } \right] + \Et\left[ ( \nu_t^{(i)} )_2 ( \nu_t^{(j)} )_2 \mathds{ 1 }_{ A_i^c } \right] \right\} \nonumber \\
&\leq N \epsilon \sum_{ j=1 }^N \sum_{ i \neq j } \Et\left[ \nu_t^{(i)} ( \nu_t^{(j)} )_2 \mathds{ 1 }_{ A_i } \right] + N^3 \sum_{ j=1 }^N \sum_{ i \neq j } \Et\left[ \nu_t^{(j)} \mathds{ 1 }_{ A_i^c } \right] \nonumber \\
&\leq N^2 ( N )_2 \epsilon \Et[ c_N( t ) ] + N^4 \sum_{ i = 1 }^N \Prob\left[ \nu_t^{(i)} > N \epsilon \mid \mathcal{F}_{t-1} \right] . \label{DN_part_3}
\end{align}
Substituting \eqref{markovs_ineq} into \eqref{DN_part_3} yields
\begin{equation}
\sum_{ j=1 }^N \sum_{ i \neq j } \Et\left[ ( \nu_t^{(i)} )_2 ( \nu_t^{(j)} )_2 \right] \leq N^4 \{ 1 + O( N^{ -1 } ) \} \Bigg( \epsilon + \frac{ b_N }{ \epsilon^3 } \Bigg) \Et[ c_N( t ) ] , \label{DN_part_4}
\end{equation}
and substituting \eqref{DN_part_4} into \eqref{DN_part_2} gives
\begin{equation}
\frac{ 1 }{ N^2 ( N )_2 } \sum_{ j=1 }^N \sum_{ i \neq j } \Et\left[ ( \nu_t^{(i)} )_2 ( \nu_t^{(j)} )^2 \right]
\leq \Bigg[ \{ 1 + O( N^{ -1 } ) \} \Big( \epsilon + \frac{ b_N }{ \epsilon^3 } \Big) + \frac{ 1 }{ N } \Bigg] \Et[ c_N( t ) ] . \label{DN_last}
\end{equation}
Combining \eqref{DN_part_1} and \eqref{DN_last}, we have that
\begin{equation*}
\Et[ D_N(t) ] \leq \Bigg[ \{ 1 + O( N^{ -1 } ) \} \Bigg( \epsilon + \frac{ b_N }{ \epsilon^3 } \Bigg) + \frac{ 3 }{ N } + b_N \Bigg] \Et[ c_N(t) ] .
\end{equation*}
Finally, invoking \cite[Lemma 2]{koskela2018} twice gives
\begin{align*}
\E\left[ \sum_{ r = \tau_N( s ) + 1 }^{ \tau_N( t ) } D_N( r ) \right] 
&= \E\left[ \sum_{ r = \tau_N( s ) + 1 }^{ \tau_N( t ) } \E_r[ D_N( r ) ] \right] \\
&\leq \Bigg[ \{ 1 + O( N^{ -1 } ) \} \Bigg( \epsilon + \frac{ b_N }{ \epsilon^3 } \Bigg) + \frac{ 3 }{ N } + b_N \Bigg] 
\E\left[ \sum_{ r = \tau_N( s ) + 1 }^{ \tau_N( t ) } c_N( r ) \right] \\
&\leq \Bigg[ \{ 1 + O( N^{ -1 } ) \} \Bigg( \epsilon + \frac{ b_N }{ \epsilon^3 } \Bigg) + \frac{ 3 }{ N } + b_N \Bigg] ( t - s + 1 )\\
& \underset{N\to\infty}{\longrightarrow} \epsilon ( t - s + 1 ),
\end{align*}
and recalling that $\epsilon > 0$ was arbitrary concludes the proof.
\end{proof}

For Lemma \ref{lem:removeass4}, we introduce the quantity $p_{\xi\eta}(t)$. For any fixed $n$ and $N$, $p_{\xi\eta}(t)$ is the time $t-1$ conditional transition probability of the genealogical process from $\xi$ to $\eta$, where $\xi$ and $\eta$ are partitions of $\{1,\dots,n\}$. The transition probability is non-zero only when $\eta =\xi$ or $\eta$ can be obtained from $\xi$ by merging some blocks.
Let $b_j\, (j=1,\dots,|\xi|)$ denote the number of blocks in $\xi$ that merged to form block $j$ of $\eta$.

\begin{lemma}\label{lem:removeass4}
\begin{equation}\label{eq:removeass4}
p_{ \xi \xi }( t ) \geq 1 - B_{ | \xi |  } \{ 1 + O( N^{ -1 } ) \} D_N( t ) -  \binom{ | \xi | }{ 2 } \{ 1 + O( N^{-1} ) \} c_N( t ) ,
\end{equation}
for a constant $B_{ | \xi |  } > 0$ increasing in $|\xi|$ that does not depend on $N$.
\end{lemma}

\begin{proof}
Let $\kappa_i = |\{ j : b_j = i \}|$ denote the multiplicity of mergers of size $i$, with the slight abuse of terminology that $\kappa_1$ counts non-merger events.
In particular, we have that $\kappa_1 + 2 \kappa_2 + \cdots + | \xi | \kappa_{ | \xi | } = | \xi |$.
Now
\begin{equation*}
p_{ \xi \xi }( t ) = 1 - \frac{ 1 }{ ( N )_{ | \xi | } } \sum_{ k = 1 }^{ | \xi | - 1 } \sum_{ \substack{ b_1 \geq \ldots \geq b_k = 1 \\ b_1 + \ldots + b_k = | \xi | } }^{ | \xi | } \frac{ | \xi |! }{ \prod_{ j = 1 }^{ | \xi | } ( j ! )^{ \kappa_j } \kappa_j ! } \sum_{ \substack{ i_1 \neq \ldots \neq i_k = 1 \\ \text{all distinct} } }^N( \nu_t^{ ( i_1 ) } )_{ b_1 } \ldots ( \nu_t^{ ( i_k ) } )_{ b_k },
\end{equation*}
because the right hand side subtracts the probabilities of all possible merger events.
See \cite[Eq (11)]{fu2006} for the combinatorial factor, which gives the number of partitions of a sequence of length $|\xi|$  having $\kappa_j$ subsequences of length $j$ for each $j$.
The omitted $k = | \xi |$ summand would correspond to the probability of an identity transition.
The non-increasing ordering of $( b_1, \ldots, b_k )$ in the sum is arbitrary, but without loss of generality: choosing any ordering of the same set of merger sizes would give the same result.

Firstly, we separate out the $k = | \xi | - 1$ term, which covers isolated binary mergers, and note that in that case the only possible $b$-vector is $(2, 1, \ldots, 1)$, for which
\begin{equation*}
\frac{ | \xi |! }{ \prod_{ j = 1 }^{ | \xi | } ( j ! )^{ \kappa_j } \kappa_j ! } = \frac{ | \xi |! }{ 2 ! ( | \xi | - 2 ) ! } = \binom{ | \xi | }{ 2 }
\end{equation*}
and a multinomial expansion argument yields
\begin{equation*}
\sum_{ i_1 \neq \ldots \neq i_{ | \xi | - 1 } = 1 }^N ( \nu_t^{ ( i_1 ) } )_2 \nu_t^{ ( i_2 ) } \ldots \nu_t^{ ( i_{ | \xi | - 1 } ) }
\leq N^{ | \xi | - 2 } \sum_{ i = 1 }^N ( \nu_t^{ ( i ) } )_2 .
\end{equation*}
Thus
\begin{align*}
p_{ \xi \xi }( t ) &\geq 1 - \binom{ | \xi | }{ 2 } \frac{ 1 + O( N^{ -1 } ) }{ ( N )_2 } \sum_{ i = 1 }^N ( \nu_t^{ ( i ) } )_2 \\
&\qquad - \frac{ 1 }{ ( N )_{ | \xi | } } \sum_{ k = 1 }^{ | \xi | - 2 } \sum_{ \substack{ b_1 \geq \ldots \geq b_k = 1 \\ b_1 + \ldots + b_k = | \xi | } }^{ | \xi | } \frac{ | \xi |! }{ \prod_{ j = 1 }^{ | \xi | } ( j ! )^{ \kappa_j } \kappa_j ! } \sum_{ \substack{ i_1 \neq \ldots \neq i_k = 1 \\ \text{all distinct} } }^N( \nu_t^{ ( i_1 ) } )_{ b_1 } \ldots ( \nu_t^{ ( i_k ) } )_{ b_k }.
\end{align*}
For the other summands, we have
\begin{equation*}
\frac{ | \xi |! }{ \prod_{ j = 1 }^{ | \xi | } ( j ! )^{ \kappa_j } \kappa_j ! } \leq | \xi | !
\end{equation*}
and, similarly to \cite[Lemma 1, Case 3]{koskela2018},
\begin{align*}
\sum_{ \substack{ i_1 \neq \ldots \neq i_k = 1 \\ \text{all distinct} } }^N &( \nu_t^{ ( i_1 ) } )_{ b_1 } \ldots ( \nu_t^{ ( i_k ) } )_{ b_k } \leq \sum_{ i = 1 }^N ( \nu_t^{ ( i ) } )_2 \Bigg( N^{ | \xi | - 2 } - \sum_{ \substack{ j_1 \neq \ldots \neq j_{ | \xi | - 2 } = 1 \\ \text{all distinct and } \neq i } }^N \nu_t^{ ( j_1 ) } \ldots \nu_t^{ ( j_{ | \xi | - 2 } ) } \Bigg) \\
&\leq \sum_{ i = 1 }^N ( \nu_t^{ ( i ) } )_2 \Bigg\{ N^{ | \xi | - 2 } - ( N - \nu_t^{ ( i ) } )^{ | \xi | - 2 } + \binom{ | \xi | - 2 }{ 2 } \sum_{ j \neq i } ( \nu_t^{ ( j ) } )^2 \Bigg( \sum_{ k \neq i } \nu_t^{ ( k ) } \Bigg)^{ | \xi | - 4 } \Bigg\} \\
&\leq \sum_{ i = 1 }^N ( \nu_t^{ ( i ) } )_2 \Bigg\{ ( | \xi | - 2 ) \nu_t^{ ( i ) } N^{ | \xi | - 3 } + \binom{ | \xi | - 2 }{ 2 } \sum_{ j \neq i } ( \nu_t^{ ( j ) } )^2 N^{ | \xi | - 4 } \Bigg\},
\end{align*}
where the last step uses $(N - x)^b \geq N^b - b x N^{ b - 1 }$ for $x \leq N$, $b \geq 0$.
Overall
\begin{multline*}
p_{ \xi \xi }( t ) \geq 1 - \binom{ | \xi | }{ 2 } \frac{ 1 + O( N^{ -1 } ) }{ ( N )_2 } \sum_{ i = 1 }^N ( \nu_t^{ ( i ) } )_2 \\
{} - \frac{ 1 }{ ( N )_{ | \xi | } } \sum_{ k = 1 }^{ | \xi | - 2 } \sum_{ \substack{ b_1 \geq \ldots \geq b_k = 1 \\ b_1 + \ldots + b_k = | \xi | } }^{ | \xi | } | \xi |! \sum_{ i = 1 }^N ( \nu_t^{ ( i ) } )_2 \Bigg\{ ( | \xi | - 2 ) \nu_t^{ ( i ) } N^{ | \xi | - 3 } + \binom{ | \xi | - 2 }{ 2 } \sum_{ j \neq i } ( \nu_t^{ ( j ) } )^2 N^{ | \xi | - 4 } \Bigg\}.
\end{multline*}
The summand in the third term depends neither on $k$ nor on $b_1, \ldots, b_k$, and the number of terms in those sums is bounded above by $\gamma_{ | \xi | - 2 } ( | \xi | - 2 )$, where $\gamma_n$ is the number of integer partitions of $n$.
By \cite[Section 2]{hardy1918}, $\gamma_n < K e^{ 2 \sqrt{ 2 n } } / n$ for a constant $K > 0$ independent of $n$.
Thus, for $|\xi| > 2$,
\begin{align*}
p_{ \xi \xi }( t ) \geq {}& 1 - \binom{ | \xi | }{ 2 } \frac{ 1 + O( N^{ -1 } ) }{ ( N )_2 } \sum_{ i = 1 }^N ( \nu_t^{ ( i ) } )_2 \\
&\qquad {}- K e^{ 2 \sqrt{ 2 ( | \xi | - 2 ) } } | \xi |! \binom{ | \xi | -1 }{ 2 } \frac{ N^{ | \xi | - 3 } }{ ( N )_{ | \xi | } } \sum_{ i = 1 }^N ( \nu_t^{ ( i ) } )_2 \Bigg\{ \nu_t^{ ( i ) } + \frac{ 1 }{ N } \sum_{ j \neq i } ( \nu_t^{ ( j ) } )^2 \Bigg\} \\
={} & 1 - \binom{ | \xi | }{ 2 } \{ 1 + O( N^{ -1 } ) \} c_N( t ) - B_{ | \xi | } \{ 1 + O( N^{ -1 } ) \} D_N( t ) ,
\end{align*}
where $B_{ | \xi | } > 0$ depends on $| \xi |$ but not on $N$. 
When $|\xi| \leq 2$, there are no higher order interactions and the result is immediate.
\end{proof}
Using \eqref{eq:removeass4} in place of
the lower bound used in the proof of \cite[Theorem 1]{koskela2018} facilitates a modification of that proof such that \eqref{eq:oldass4} is not needed.
Details of this argument are in the following subsection.

\subsection{Proof of Theorem \ref{thm:generalconv} without Assumption \eqref{eq:oldass4}}\label{sec:proof}

The proof of \cite[Theorem 1]{koskela2018} proceeds in three parts.
The first is a vanishing upper bound on finite-dimensional distributions of the genealogical process when the path of the process involves either multiple simultaneous mergers or any merger involving more than two particles.
The second is showing that the finite-dimensional distributions of the $n$-coalescent upper bound those of the genealogical process in the limit $N \rightarrow \infty$ when the path of the genealogy consists of only isolated binary mergers.
The final piece is a similar lower bound, which establishes convergence of the finite-dimensional distributions.
Only the third part makes use of assumption \eqref{eq:oldass4}.
Hence, it suffices to show that Lemma \ref{lem:removeass4} can be used to obtain the same lower bound without making use of \eqref{eq:oldass4}.

\begin{proof}
Let $\chi^\star_d$ be the conditional transition probability of a transition from state $\eta_{d-1}$ to state $\eta_d$ at times $\tau_N(t_{d-1})$ and $\tau_N(t_d)$ respectively, conditional on the offspring counts between those times $\nu^{(1:N)}_{\tau_N(d-1) +1} , \dots, \nu^{(1:N)}_{\tau_N(d)}$. 
This transition can happen via any valid path of merger events, but we restrict to paths involving binary mergers only, and denote by $\chi_d$ the conditional transition probability subject to this restriction.
Considering the Proof of Theorem 1 in \cite{koskela2018}, the derivation of an upper bound on $\chi_d$ holds without modification in our context; while the first step in the derivation of a lower bound which can be found on page 14 of that work involves the application of its Lemma 1 to bound $\chi_d$ from below.
Instead, we apply Lemma \ref{lem:removeass4} to obtain
\begin{align*}
\chi_d &\geq \sum_{ s_1 < \ldots < s_{ \alpha } = \tau_N( t_{ d - 1 } ) + 1 }^{ \tau_N( t_d ) } ( \tilde{ Q }^{ \alpha } )_{ \eta_{ d - 1 } \eta_d } \Bigg( \prod_{ r = 1 }^{ \alpha } \Bigg[ c_N( s_r ) - \binom{ n - 2 }{ 2 } \{ 1+ O( N^{ -1 } ) \} D_N( s_r ) \Bigg] \Bigg) \\
&\qquad\qquad\phantom{\geq} \times \prod_{ \substack{ r = \tau_N( t_{ d - 1 } ) + 1 \\ r \neq s_1, \ldots, r \neq s_{ \alpha } } }^{ \tau_N( t_d ) } \Bigg[ 1 - B_n \{ 1 + O( N^{ -1 } ) \} D_N( r ) \\
&\qquad\qquad\qquad\phantom{ \geq \times \prod_{ \substack{ r = \tau_N( t_{ d - 1 } ) + 1 \\ r \neq s_1, \ldots, r \neq s_{ \alpha } } }^{ \tau_N( t_d ) } \Bigg[ 1 } - \binom{ | \eta_{ d - 1 } | - | \{ i : s_i < r \} | }{ 2 } \{ 1 + O( N^{ -1 } ) \} c_N( r ) \Bigg] .
\end{align*}
Here $\tilde{Q}$ is the matrix obtained from the generator $Q$ of Kingman's $n$-coalescent (see Definition \ref{def:kingman}) by setting the diagonal entries to 0.
The number of pair-merger steps required to transition from $\eta_{d-1}$ to $\eta_d$ is $\alpha = |\eta_{d-1}| - |\eta_d|$. The sequences $s_1,\dots,s_\alpha$ denote the times at which these pair-mergers happen. 
At the remaining times $r$ the partition is unchanged, and the bound in Lemma \ref{lem:removeass4} has been applied to the one-step transition probabilities corresponding to these identity transtions. The constant $B_n$ is that appearing in Lemma \ref{lem:removeass4}, where we replace $|\eta_d|$ by its upper bound $n$.
A sum over an index vector of length zero should be interpreted as the identity operator here and in the following.

The rest of the proof proceeds as in \cite{koskela2018}, albeit from this modified initial lower bound.
A multinomial expansion of the product on the second line yields 
\begin{align*}
\chi_d \geq {}& \sum_{ \beta = 0 }^{ \tau_N( t_d ) - \tau_N( t_{ d - 1 } ) - \alpha } ( \tilde{ Q }^{ \alpha } )_{ \eta_{ d - 1 } \eta_d } \sum_{\substack{ ( \lambda, \mu ) \in \Pi_2( [ \alpha + \beta ] ) :\\ | \lambda | = \alpha }} \{ 1 + O( N^{ -1 } ) \}^{ \beta } \\
&\times \sum_{ s_1 < \ldots < s_{ \alpha + \beta } = \tau_N( t_{ d - 1 } ) + 1 }^{ \tau_N( t_d ) } \Bigg( \prod_{ r \in \lambda } \Bigg[ c_N( s_r ) - \binom{ n - 2 }{ 2 } \{ 1 + O( N^{ -1 } ) \} D_N( s_r ) \Bigg] \Bigg)\\
&\times \prod_{ r \in \mu } \Bigg\{ - \binom{ | \eta_{ d - 1 } | - | \{ i \in \lambda : i < r \} | }{ 2 } c_N( s_r ) - B_n D_N( s_r ) \Bigg\}
\end{align*}
where $\Pi_i([n])$ denotes the set of partitions of $\{1, \dots,n\}$ into exactly $i$ blocks.
Expanding the product over $\lambda$ gives
\begin{align*}
\chi_d &\geq \sum_{ \beta = 0 }^{ \tau_N( t_d ) - \tau_N( t_{ d - 1 } ) - \alpha } ( \tilde{ Q }^{ \alpha } )_{ \eta_{ d - 1 } \eta_d } \sum_{\substack{ ( \lambda, \mu, \pi ) \in \Pi_3( [ \alpha + \beta ] ) :\\ | \mu | = \beta }} \binom{ n - 2 }{ 2 }^{ | \pi | } ( -1 )^{ | \pi | } \{ 1 + O( N^{ -1 } ) \}^{ \beta + | \pi | } \\
&\phantom{\geq} \times \sum_{ s_1 < \ldots < s_{ \alpha + \beta } = \tau_N( t_{ d - 1 } ) + 1 }^{ \tau_N( t_d ) } \Bigg\{ \prod_{ r \in \lambda } c_N( s_r ) \Bigg\} \Bigg\{ \prod_{ r \in \pi }  D_N( s_r ) \Bigg\} \\
&\phantom{\geq} \times \prod_{ r \in \mu } \Bigg\{ - \binom{ | \eta_{ d - 1 } | - | \{ i \in \lambda \cup \pi : i < r \} | }{ 2 } c_N( s_r ) - B_n D_N( s_r ) \Bigg\}
\end{align*}
and expanding the product over $\mu$ results in
\begin{align*}
\chi_d &\geq \sum_{ \beta = 0 }^{ \tau_N( t_d ) - \tau_N( t_{ d - 1 } ) - \alpha } ( \tilde{ Q }^{ \alpha } )_{ \eta_{ d - 1 } \eta_d } \sum_{\substack{ ( \lambda, \mu, \pi, \sigma ) \in \Pi_4( [ \alpha + \beta ] ) :\\ | \mu | + | \sigma | = \beta }} B_n^{ | \sigma | } \binom{ n - 2 }{ 2 }^{ | \pi | } ( -1 )^{ | \pi | + | \sigma | } \\
&\phantom{\geq} \times \{ 1 + O( N^{ -1 } ) \}^{ \beta + | \pi | } \Bigg\{ \prod_{ r \in \mu } - \binom{ | \eta_{ d - 1 } | - | \{ i \in \lambda \cup \pi : i < r \} | }{ 2 } \Bigg\} \\
&\phantom{\geq} \times \sum_{ s_1 < \ldots < s_{ \alpha + \beta } = \tau_N( t_{ d - 1 } ) + 1 }^{ \tau_N( t_d ) } \Bigg\{ \prod_{ r \in \lambda \cup \mu } c_N( s_r ) \Bigg\} \prod_{ r \in \pi \cup \sigma }  D_N( s_r ) .
\end{align*}
Via a further multinomial expansion, the lower bound for the $k$-step transition probability can be written as
\begin{align*}
\lim_{ N \rightarrow \infty } \E\left[ \prod_{ d = 1 }^k \chi_d \right]
&\geq \lim_{ N \rightarrow \infty } \E\Bigg[ \sum_{ \beta_1 = 0 }^{ \infty } \mkern-6mu \ldots \mkern-6mu \sum_{ \beta_k = 0 }^{ \infty } \sum_{\substack{ ( \lambda_1, \mu_1, \pi_1, \sigma_1 ) \in \Pi_4( [ \alpha_1 + \beta_1 ] ):\\  | \mu_1 | + | \sigma_1 | = \beta_1 }} \ldots 
\sum_{\substack{ ( \lambda_k, \mu_k, \pi_k, \sigma_k ) \in \Pi_4( [ \alpha_k + \beta_k ] ) :\\ | \mu_k | + | \sigma_k | = \beta_k }} \\
&\phantom{\geq} B_n^{ \sum_{ d = 1 }^k | \sigma_d | } \binom{ n - 2 }{ 2 }^{ \sum_{ d = 1 }^k| \pi_d | }
( -1 )^{ \sum_{ d = 1 }^k | \pi_d | + | \sigma_d | }  \{ 1 + O( N^{ -1 } ) \}^{ | \bm{ \beta } | + \sum_{ d = 1 }^k | \pi_d | } \\
&\phantom{\geq} \times \Bigg\{ \prod_{ d = 1 }^k ( \tilde{ Q }^{ \alpha_d } )_{ \eta_{ d - 1 } \eta_d } \prod_{ r \in \mu_d } - \binom{ | \eta_{ d - 1 } | - | \{ i \in \lambda_d \cup \pi_d : i < r \} | }{ 2 } \Bigg\} \\
&\phantom{\geq} \times \sum_{ s_1^{ ( 1 ) } < \ldots < s_{ \alpha_1 + \beta_1 }^{ ( 1 ) } = \tau_N( t_0 ) + 1 }^{ \tau_N( t_1 ) } \ldots \sum_{ s_1^{ ( k ) } < \ldots < s_{ \alpha_k + \beta_k }^{ ( k ) } = \tau_N( t_{ k - 1 } ) + 1 }^{ \tau_N( t_k ) } \\
&\phantom{\geq} \prod_{ d = 1 }^k \mathds{ 1 }_{ \{ \tau_N( t_d ) - \tau_N( t_{ d - 1 } ) \geq \alpha_d + \beta_d \} } \Bigg\{ \prod_{ r \in \lambda_d \cup \mu_d } c_N( s_r^{ ( d ) } ) \Bigg\} \prod_{ r \in \pi_d \cup \sigma_d }  D_N( s_r^{ ( d ) } ) \Bigg].
\end{align*}
An argument completely analogous to that in \cite[Appendix]{koskela2018} shows that passing the expectation and the limit through the infinite sums is justified, whereupon the contribution of terms with $ \sum_{ d = 1 }^k ( | \pi_d | + | \sigma_d | ) > 0$ vanishes.
To see why, follow the argument used to show that the contribution of multiple merger trajectories vanishes in the corresponding upper bound in \cite{koskela2018}.
That leaves
\begin{align}
\lim_{ N \rightarrow \infty } \E\left[ \prod_{ d = 1 }^k \chi_d \right] 
&\geq \sum_{ \beta_1 = 0 }^{ \infty } \ldots \sum_{ \beta_k = 0 }^{ \infty } \sum_{\substack{ ( \lambda_1, \mu_1 ) \in \Pi_2( [ \alpha_1 + \beta_1 ] ) :\\ | \mu_1 | = \beta_1 }} \ldots \sum_{\substack{ ( \lambda_k, \mu_k ) \in \Pi_2( [ \alpha_k + \beta_k ] ) :\\ | \mu_k | = \beta_k }} \nonumber \\
&\phantom{\geq} \Bigg\{ \prod_{ d = 1 }^k ( \tilde{ Q }^{ \alpha_d } )_{ \eta_{ d - 1 } \eta_d } \prod_{ r \in \mu_d } - \binom{ | \eta_{ d - 1 } | - | \{ i \in \lambda_d \cup \pi_d : i < r \} | }{ 2 } \Bigg\} \nonumber \\
&\phantom{\geq} \times \lim_{ N \rightarrow \infty } \E\Bigg[ \sum_{ s_1^{ ( 1 ) } < \ldots < s_{ \alpha_1 + \beta_1 }^{ ( 1 ) } = \tau_N( t_0 ) + 1 }^{ \tau_N( t_1 ) } \ldots \sum_{ s_1^{ ( k ) } < \ldots < s_{ \alpha_k + \beta_k }^{ ( k ) } = \tau_N( t_{ k - 1 } ) + 1 }^{ \tau_N( t_k ) } \nonumber \\
&\phantom{\geq} \prod_{ d = 1 }^k \mathds{ 1 }_{ \{ \tau_N( t_d ) - \tau_N( t_{ d - 1 } ) \geq \alpha_d + \beta_d \} } \Bigg\{ \prod_{ r \in \lambda_d \cup \mu_d } c_N( s_r^{ ( d ) } ) \Bigg\} \Bigg]. \label{eq1}
\end{align}
Recall \cite[Eq (11)]{koskela2018}:
\begin{align*}
\sum_{\substack{ ( \lambda, \mu ) \in \Pi_2( [ \alpha + \beta ] ) :\\ | \mu | = \beta }} ( \tilde{ Q }^{ \alpha } )_{ \eta_{ d - 1 } \eta_d } \prod_{ r \in \mu } - \binom{ | \eta_{ d - 1 } | - | \{ i \in \lambda \cup \pi : i < r \} | }{ 2 } = ( Q^{ \alpha + \beta } )_{ \eta_{ d - 1 } \eta_d } .
\end{align*}
Applying this $k$ times in \eqref{eq1} yields
\begin{align*}
\lim_{ N \rightarrow \infty } \E\left[ \prod_{ d = 1 }^k \chi_d \right]
&\geq \sum_{ \beta_1 = 0 }^{ \infty } \ldots \sum_{ \beta_k = 0 }^{ \infty } \Bigg\{ \prod_{ d = 1 }^k ( Q^{ \alpha_d + \beta_d } )_{ \eta_{ d - 1 } \eta_d } \Bigg\} \\
&\phantom{\geq} \times \lim_{ N \rightarrow \infty } \E\Bigg[ \Bigg( \prod_{ d = 1 }^k \mathds{ 1 }_{ \{ \tau_N( t_d ) - \tau_N( t_{ d - 1 } ) \geq \alpha_d + \beta_d \} } \Bigg) \sum_{ s_1^{ ( 1 ) } < \ldots < s_{ \alpha_1 + \beta_1 }^{ ( 1 ) } = \tau_N( t_0 ) + 1 }^{ \tau_N( t_1 ) } \\
&\phantom{\geq \times \lim_{ N \rightarrow \infty } \E\Bigg[} \ldots \sum_{ s_1^{ ( k ) } < \ldots < s_{ \alpha_k + \beta_k }^{ ( k ) } = \tau_N( t_{ k - 1 } ) + 1 }^{ \tau_N( t_k ) } \prod_{ d = 1 }^k \prod_{ r \in \lambda_d \cup \mu_d } c_N( s_r^{ ( d ) } ) \Bigg].
\end{align*}
We now apply equations (14) and (15), respectively, of \cite{koskela2018},  to those terms with a negative ($|\bm{\beta}|$ odd) and positive ($|\bm{\beta}|$ even) sign, respectively, and obtain
\begin{align*}
\lim_{ N \rightarrow \infty } \E\left[ \prod_{ d = 1 }^k \chi_d \right]
&\geq \sum_{ \beta_1 = 0 }^{ \infty } \ldots \sum_{ \beta_k = 0 }^{ \infty } \Bigg\{ \prod_{ d = 1 }^k ( Q^{ \alpha_d + \beta_d } )_{ \eta_{ d - 1 } \eta_d } \frac{ ( t_d - t_{ d - 1 } )^{ \alpha_d + \beta_d } }{ ( \alpha_d + \beta_d ) ! } \Bigg\} \\
&\phantom{\geq} \times \lim_{ N \rightarrow \infty } \E\left[ \prod_{ d = 1 }^k \mathds{ 1 }_{ \{ \tau_N( t_d ) - \tau_N( t_{ d - 1 } ) \geq \alpha_d + \beta_d \} } \right].
\end{align*}
An invocation of \cite[Eq (16)]{koskela2018} concludes the proof.
\end{proof}

\section{Illustrative Applications}\label{sec:applications}
\subsection{Resampling with stochastic roundings}
\begin{definition}\label{defn:stochround}
 Let $X=(X_1,\dots,X_N)$ be a $\mathbb{R}_+^N$-valued random variable. Then $Y=(Y_1,\dots,Y_N) \in \mathbb{N}^N$ is a \emph{stochastic rounding} of $X$ if each element $Y_i$ takes values
\begin{equation*}
Y_i \mid X =
\begin{cases}
 \lfloor X_i \rfloor & \text{with probability } 1- X_i+ \lfloor X_i \rfloor \\
  \lfloor X_i \rfloor +1 & \text{with probability } X_i- \lfloor X_i \rfloor .
\end{cases}
\end{equation*}
\end{definition}

By construction, $\E[Y_i \mid X] = X_i$ for each $i$. Taking $X$ to be $N$ times the vector of particle weights, we can therefore use stochastic rounding for the \textsc{resample} procedure in Algorithm \ref{alg:SMC}, under the further constraint that $Y_1 + \dots + Y_N = N$. Several ways to enforce this constraint have been proposed, including systematic resampling \cite{carpenter1999, whitley1994}, residual resampling with stratified or systematic residuals \cite{whitley1994}, the branching system of \cite{crisan1997}, and the Srinivasan sampling process resampling introduced in \cite{gerber2017}.

\begin{corollary}\label{thm:stochrounding}
Consider an SMC algorithm using any stochastic rounding as its resampling scheme, such that the standing assumption is satisfied.
Assume that there exists a constant $a\in [1,\infty)$ such that for all $x, x^\prime, t$,
\begin{equation}\label{eq:gq_bounds_sr}
\frac{1}{a} \leq g_t(x, x^\prime) \leq a . 
\end{equation}
Assume that $\Prob[ \tau_N(t) = \infty] =0$ for all finite $t$.
Let $(G_t^{(n,N)})_{t\geq0}$ denote the genealogy of a random sample of $n$ terminal particles from the output of the algorithm when the total number of particles used is $N$. Then, for any fixed $n$, the time-scaled genealogy $(G_{\tau_N(t)}^{(n,N)})_{t\geq0}$ converges to Kingman's $n$-coalescent as $N\to \infty$, in the sense of finite-dimensional distributions.
\end{corollary}

Condition \eqref{eq:gq_bounds_sr} is strong, but is widespread in the literature.
By contrast, the assumption that the time-scale does not explode is a technicality, which holds under an easier-to-verify condition when the transition densities are bounded above and below (see Lemma~\ref{thm:SR_nontriviality}). We conjecture that in general bounded transition densities are not necessary in the case of stochastic rounding.

\begin{remark}
In a similar vein to \cite[Remark 3]{koskela2018}, if we consider the weight vector as fixed, the time scale induced by stochastic rounding is slower than that induced by multinomial resampling. Details are given in Appendix~\ref{app:timescales}.
\end{remark}

\begin{remark}
Since every stochastic rounding has the same marginal distributions and the first moment of \eqref{eq:coalescence_rate} depends only on marginal family size distributions, the expected coalescence rate is the same whichever stochastic rounding is used for resampling. Thus the time-scale on which the $n$-coalescent is recovered is equal in expectation for every such scheme.
\end{remark}

\begin{proof}
Using the forward-time Markov property of SMC, and the associated conditional dependence graph, for each $N$ we establish a sequence of $\sigma$-algebras
\begin{equation}\label{eq:defn_Ht}
\mathcal{H}_t := \sigma(X_{t-1}^{(1:N)}, X_t^{(1:N)}, w_{t-1}^{(1:N)}, w_t^{(1:N)} )
\end{equation}
such that $\nu_t^{(1:N)}$ is conditionally independent of the filtration $\mathcal{F}_{t-1}$ given $\mathcal{H}_t$. The full D-separation argument is presented in Appendix~\ref{app:dseparation}.

Defining the family sizes $\nu_t^{(i)} = |\{ j : a_t^{(j)} = i \}|$ as functions of $a_t^{(1:N)}$, we have the almost sure constraint $\nu_t^{(i)} \in \{\flnw, \flnw +1\}$.  Denote $p_0^{(i)} := \Prob[ \nu_t^{(i)} = \flnw \mid \mathcal{H}_t ]$ and $p_1^{(i)} := \Prob[ \nu_t^{(i)} = \flnw +1 \mid \mathcal{H}_t ] = 1-p_0^{(i)}$. 

We obtain the following upper bounds, using the almost sure bounds $w_t^{(i)} \leq a^2/N$ which follow from \eqref{eq:gq_bounds_sr} along with the form of the weights in Algorithm \ref{alg:SMC}:
\begin{align*}
\E[(\nu_t^{(i)})_3 \mid \mathcal{H}_t] &= p_0^{(i)} (\flnw)_3 + p_1^{(i)} (\flnw + 1)_3 \\
&= \flnw (\flnw -1) \{ p_0^{(i)} (\flnw -2) + p_1^{(i)} (\flnw +1) \} \\
&= \flnw (\flnw -1) \{ \flnw (p_0^{(i)} + p_1^{(i)}) -2 p_0^{(i)} + p_1^{(i)} \} \\
&= \flnw (\flnw -1) \{ \flnw -2 p_0^{(i)} + p_1^{(i)} \} \\
&\leq a^2 (a^2 -1) (a^2 -0 +1) \1{\flnw \geq 2} \\
&\leq (a^2 +1)^3 \1{\flnw \geq 2} .
\end{align*}
We also have the lower bounds
\begin{align*}
\E[(\nu_t^{(i)})_2 \mid \mathcal{H}_t] &= p_0^{(i)} (\flnw)_2 + p_1^{(i)} (\flnw + 1)_2 \\
&= \flnw \{ p_0^{(i)} (\flnw -1) + p_1^{(i)} (\flnw +1) \} \\
&= \flnw \{ \flnw (p_0^{(i)} + p_1^{(i)}) - p_0^{(i)} + p_1^{(i)} \} \\
&= \flnw \{ \flnw - p_0^{(i)} + p_1^{(i)} \} \\
&\geq 2 (2 -1 +0) \1{\flnw \geq 2} = 2 \1{\flnw \geq 2} .
\end{align*}
Applying the tower property and conditional independence,
\begin{align*}
\frac{1}{(N)_2} \sum_{i=1}^N \Et [( \nu_t^{(i)} )_2 ] 
&= \frac{1}{(N)_2} \Et\left[ \sum_{i=1}^N \E\left[ (\nu_t^{(i)})_2 \mid \mathcal{H}_t, \mathcal{F}_{t-1} \right] \right] \\
&= \frac{1}{(N)_2} \Et\left[ \sum_{i=1}^N \E\left[ (\nu_t^{(i)})_2 \mid \mathcal{H}_t \right] \right]
\geq \frac{1}{(N)_2} 2 \Et\left[ |\{i:\flnw\geq 2\}| \right] 
\end{align*}
and similarly
\begin{align*}
\frac{1}{(N)_3} \sum_{i=1}^N \Et [( \nu_t^{(i)} )_3 ] 
&\leq \frac{1}{(N)_3}  (a^2+1)^3 \Et\left[ |\{i:\flnw\geq 2\}| \right] \\
&\leq b_N  \frac{1}{(N)_2} \sum_{i=1}^N \Et [( \nu_t^{(i)} )_2 ]
\end{align*}
where
\begin{equation*}
b_N := \frac{1}{N-2}\frac{(a^2+1)^3}{2} \underset{N\to\infty}{\longrightarrow} 0
\end{equation*}
is independent of $\mathcal{F}_\infty$, satisfying \eqref{eq:mainthmcond}.
The result follows by applying Theorem \ref{thm:generalconv}.
\end{proof}

\subsection{Conditional sequential Monte Carlo updates}

Conditional SMC differs from Algorithm \ref{alg:SMC} in that one predetermined trajectory is conditioned to survive all of the resampling steps. 
We refer to this sequence of positions as the \emph{immortal trajectory}, and the \emph{immortal particle} is the particle in a particular generation that is part of the immortal trajectory.
Conditional SMC was introduced as a component of the particle Gibbs algorithm \cite{andrieu2010} but has found somewhat wider application in fields as diverse as smoothing \cite{jacob19,shestopaloff19} and optimization \cite[Chapter 6]{finke15}. 

In particle Gibbs, the immortal trajectory $x_{0:T}^\star$ at each time step is sampled from the output of the previous conditional SMC run. It is therefore important that with high probability at least two distinct lineages survive each run so that the immortal trajectory can be updated.
A single run of conditional SMC with multinomial resampling is presented in Algorithm \ref{alg:condSMC}.

\vspace{10pt}
\begin{algorithm}[H]
\DontPrintSemicolon
\KwIn{$N, T, \mu, (K_t)_{t=1}^T, (g_t)_{t=0}^T, x_{0:T}^\star, a_{0:T}^\star$}
Set $X_0^{(a_0^\star)} \gets x_0^\star$\;
\lFor{$i \in \{1,\dots,N\} \setminus a_0^\star$}{
	Sample $X_0^{(i)} \sim \mu(\cdot)$
}
\lFor{$i \in \{1,\dots, N\}$}{
		$w_{0}^{(i)} \gets  \left\{ {\sum_{j=1}^N g_0(X_0^{(j)})}\right \}^{-1}{g_0(X_0^{(i)})} $
}                
\For{$t \in \{0,\dots, T-1\}$}{
	Set $a_t^{(a_{t+1}^\star)} \gets a_t^\star$, 	$X_{t+1}^{(a_{t+1}^\star)} \gets x_{t+1}^\star$\;
	\lFor{$i \in \{1,\dots, N\} \setminus a_{t+1}^\star $}{
        Sample $a_t^{(i)} \sim \operatorname{Categorical}(\{1,\dots,N\}, w_t^{(1:N)})$
    }
	\lFor{$i \in \{1,\dots,N\}$}{
		Sample $X_{t+1}^{(i)} \sim K_{t+1}(X_t^{(a_t^{(i)})}, \cdot)$
	}
    \lFor{$i \in \{1,\dots, N\}$}{	
		$w_{t+1}^{(i)} \gets \Big\{ {\sum_{j=1}^Ng_{t+1}(X_t^{(a_t^{(j)})},X_{t+1}^{(j)}) }\Big\}^{-1} g_{t+1}(X_t^{(a_t^{(i)})},X_{t+1}^{(i)}) $
	}          
}
\caption{Conditional sequential Monte Carlo with multinomial resampling (forwards in time)}
\label{alg:condSMC}
\end{algorithm}
\vspace{10pt}

Although it is also possible to construct a conditional SMC algorithm using a low-variance resampling scheme, here we treat only the case of multinomial resampling. We believe that the result can be extended to other resampling schemes by similar arguments to those of Corollary \ref{thm:stochrounding}.

\begin{corollary}\label{thm:CSMC_newassns}
Consider a conditional SMC algorithm using multinomial resampling, such that the standing assumption is satisfied. Assume there exist constants $\varepsilon\in (0,1], a\in [1,\infty)$ and probability density $h$ such that for all $x, x^\prime, t$,
\begin{equation}\label{eq:gq_bounds_csmc}
\frac{1}{a} \leq g_t(x, x^\prime) \leq a , \quad
\varepsilon h(x^\prime) \leq q_t(x, x^\prime) \leq \frac{1}{\varepsilon} h(x^\prime) .
\end{equation}
Let $(G_t^{(n,N)})_{t\geq0}$ denote the genealogy of a random sample of $n$ terminal particles from the output of the algorithm when the total number of particles used is $N$. Then, for any fixed $n$, the time-scaled genealogy $(G_{\tau_N(t)}^{(n,N)})_{t\geq0}$ converges to Kingman's $n$-coalescent as $N\to \infty$, in the sense of finite-dimensional distributions.
\end{corollary}

Condition \eqref{eq:gq_bounds_csmc} is widespread in the SMC literature, where it is known as the \emph{strong mixing condition} \cite[Section 3.5.2]{delmoral2004}; it further strengthens the condition \eqref{eq:gq_bounds_sr} used in the case of stochastic rounding. It is to be expected that some additional control over transition densities is required in the case of multinomial resampling compared to stochastic rounding.

\begin{proof}
Define the conditioning $\sigma$-algebra $\mathcal{H}_t$ as in \eqref{eq:defn_Ht}.
We assume without loss of generality that the immortal particle takes index 1 in each generation. This significantly simplifies the notation, but the same argument holds if the immortal indices are taken to be $a_{(0:T)}^\star$ rather than $(1,\dots,1)$.

The parental indices are conditionally independent, as in standard SMC with multinomial resampling, but we have to treat $i=1$ as a special case. We have the following conditional law on parental indices
\begin{equation*}
\Prob \left[ a_t^{(i)} = a_i \mid \mathcal{H}_t \right] \propto
\begin{cases}
\1{a_i=1} &i=1 \\
w_t^{(a_i)} q_{t-1}(X_t^{(a_i)}, X_{t-1}^{(i)}) &i=2,\dots,N .
\end{cases}
\end{equation*}
The joint conditional law is therefore
\begin{equation*}
\Prob \left[ a_t^{(1:N)} = a_{1:N} \mid \mathcal{H}_t \right] \propto \1{a_1 = 1} \prod_{i=2}^N w_t^{(a_i)} q_{t-1}(X_t^{(a_i)}, X_{t-1}^{(i)}).
\end{equation*}
First we make the following observation, which follows from a balls-in-bins coupling.
Assume \eqref{eq:gq_bounds_csmc}. 
Then for any function $f:\{1,\dots,N\}^N \to \mathbb{R}$ such that (for a fixed $i$) $f(a_t^{\prime(1:N)}) \geq f(a_t^{(1:N)})$ whenever $|\{j:a_t^{\prime(j)}=i\}| \geq |\{j:a_t^{(j)}=i\}|$,
\begin{equation}\label{eq:csmc_f_bound}
\E[ f(A_{1,i}^{(1:N)}) ] 
\leq \E[ f(a_t^{(1:N)}) \mid \mathcal{H}_t ]
\leq \E[ f(A_{2,i}^{(1:N)}) ] 
\end{equation}
where the elements of $A_{1,i}^{(1:N)}, A_{2,i}^{(1:N)}$ are all mutually independent and independent of $\mathcal{F}_{\infty}$, and distributed according to
\begin{align*}
& A_{1,i}^{(j)} \sim \begin{cases}
\delta_1 \qquad & j=1 \\
\operatorname{Categorical}\left( (\varepsilon/a)^{\1{i=1} -\1{i\neq 1}} ,\dots, (\varepsilon/a)^{\1{i=N} -\1{i\neq N}} \right) & j\neq 1 
\end{cases} \\
& A_{2,i}^{(j)} \sim \begin{cases}
\delta_1 \qquad & j=1\\
\operatorname{Categorical}\left( (a/\varepsilon)^{\1{i=1} -\1{i\neq 1}} ,\dots, (a/\varepsilon)^{\1{i=N} -\1{i\neq N}} \right) & j\neq 1
 \end{cases}
\end{align*}
where the vector of probabilities is given up to a constant in the argument of Categorical distributions.
We use these random vectors to construct bounds that are independent of $\mathcal{F}_\infty$.
Also define the corresponding offspring counts $V_1^{(i)} = |\{j: A_{1,i}^{(j)}=i\}|$, $V_2^{(i)} = |\{j: A_{2,i}^{(j)}=i\}|$, for $i=1,\dots,N$, which have marginal distributions
\begin{align*}
& V_1^{(i)} \overset{d}{=} \1{i=1} + \operatorname{Binomial}\left(N-1, \frac{\varepsilon/a}{(\varepsilon/a) + (N-1)(a/\varepsilon)} \right) , \\
& V_2^{(i)} \overset{d}{=} \1{i=1} + \operatorname{Binomial}\left( N-1, \frac{a/\varepsilon}{(a/\varepsilon) + (N-1)(\varepsilon/a)} \right) .
\end{align*}
Now consider the function $f_i(a_t^{(1:N)}) := (\nu_t^{(i)})_2$. We can apply \eqref{eq:csmc_f_bound} to obtain the lower bound
\begin{align*}
\frac{1}{(N)_2} \sum_{i=1}^N \E[ (\nu_t^{(i)})_2 \mid \mathcal{H}_t ]
&\geq \frac{1}{(N)_2} \sum_{i=1}^N \E[ (V_1^{(i)})_2 ]
=  \frac{1}{(N)_2} \left[ \E[ (V_1^{(1)})_2 ] + \sum_{i=2}^N \E[ (V_1^{(i)})_2 ] \right] \\
&= \frac{1}{(N)_2} \Bigg[ \frac{(N-1)_2 (\varepsilon/a)^2}{\{(\varepsilon/a) + (N-1)(a/\varepsilon)\}^2} + \frac{2(N-1)(\varepsilon/a)}{(\varepsilon/a) + (N-1)(a/\varepsilon)}  \\
&\qquad\qquad\qquad + \sum_{i=2}^N \frac{(N-1)_2 (\varepsilon/a)^2}{\{(\varepsilon/a) + (N-1)(a/\varepsilon)\}^2} \Bigg] \\
&= \frac{1}{(N)_2} \left[ \frac{2(N-1)(\varepsilon/a)}{(\varepsilon/a) + (N-1)(a/\varepsilon)} + \sum_{i=1}^N \frac{(N-1)_2 (\varepsilon/a)^2}{\{(\varepsilon/a) + (N-1)(a/\varepsilon)\}^2} \right]
\end{align*}
using the moments of the Binomial distribution (see \cite{mosimann1962} for example) along with the identity $(X+1)_2 \equiv 2(X)_1 +(X)_2$.
This is further bounded by
\begin{align}
\frac{1}{(N)_2} \sum_{i=1}^N \E[ (\nu_t^{(i)})_2 \mid \mathcal{H}_t ]
&\geq \frac{1}{(N)_2} \left\{ \frac{2(N-1)(\varepsilon/a)}{N(a/\varepsilon)} + \frac{(N)_3 (\varepsilon/a)^2}{N^2(a/\varepsilon)^2} \right\} \nonumber\\
&= \frac{1}{N^2} \left\{\frac{2\varepsilon^2}{a^2} + \frac{(N-2)\varepsilon^4}{a^4}  \right\} . \label{eq:CSMC_cN_LB}
\end{align}
Similarly, we derive an upper bound on $f_i(a_t^{(1:N)}) := (\nu_t^{(i)})_3$, this time using the identity $(X+1)_3 \equiv 3(X)_2 +(X)_3 $:
\begin{align*}
\frac{1}{(N)_3} \sum_{i=1}^N \E[ (\nu_t^{(i)})_3 \mid \mathcal{H}_t]
&\leq \frac{1}{(N)_3} \left[ \E[ (V_2^{(1)})_3 ] + \sum_{i=2}^N \E[ (V_2^{(i)})_3 ] \right] \\
&\leq \frac{1}{(N)_3} \left[ \frac{ 3 (N-1)_2 (a/\varepsilon)^2}{\{(a/\varepsilon) + (N-1)(\varepsilon/a)\}^2} + \sum_{i=1}^N \frac{(N-1)_3 (a/\varepsilon)^3}{\{(a/\varepsilon) + (N-1)(\varepsilon/a)\}^3} \right] \\
&\leq \frac{1}{(N)_3} \left\{ \frac{3(N-1)_2 (a/\varepsilon)^2}{N^2 (\varepsilon/a)^2} + \frac{(N)_4 (a/\varepsilon)^3}{N^3 (\varepsilon/a)^3} \right\} \\
&= \frac{1}{(N)_3} \left\{ \frac{3(N-1)_2}{N^2} \frac{a^4}{\varepsilon^4} +\frac{(N)_4}{N^3} \frac{a^6}{\varepsilon^6} \right\} \\
&= \frac{1}{N^3} \left\{ \frac{3a^4}{\varepsilon^4} + \frac{(N-3) a^6}{\varepsilon^6} \right\} .
\end{align*}
We apply the tower property and conditional independence as in Corollary~\ref{thm:stochrounding}, upper bounding the ratio by
\begin{align*}
\frac{\frac{1}{(N)_3} \sum_{i=1}^N \Et[ (\nu_t^{(i)})_3 ]}{\frac{1}{(N)_2} \sum_{i=1}^N \Et[ (\nu_t^{(i)})_2 ]}
&\leq \frac{N^2}{N^3} \frac{ \frac{3a^4}{\varepsilon^4} + \frac{(N-3)a^6}{\varepsilon^6} }{ \frac{2\varepsilon^2}{a^2} + \frac{(N-2)\varepsilon^4}{a^4} }
\leq \frac{1}{N} \frac{a^6}{\varepsilon^6}\, \frac{3 + (N-3) a^2 / \varepsilon^2 }{2 + (N-2) \varepsilon^2 / a^2} \\
&\leq \frac{1}{N} \frac{a^6}{\varepsilon^6} \left\{ \frac{3}{2} + \frac{N-3}{N-2} \frac{a^4}{\varepsilon^4} \right\}
\leq \frac{1}{N} \left\{ \frac{3 a^6}{2 \varepsilon^6} + \frac{a^{10}}{\varepsilon^{10}} \right\}
=: b_N \underset{N\to\infty}{\longrightarrow} 0.
\end{align*}
Thus \eqref{eq:mainthmcond} is satisfied. It remains to show that, for $N$ sufficiently large, $\Prob[ \tau_N(t) = \infty ] =0$ for all finite $t$, a technicality which is proved in Lemma \ref{thm:CSMC_nontriviality} in Appendix~\ref{app:finiteness}. Applying Theorem \ref{thm:generalconv} gives the result.
\end{proof}

\section*{Acknowledgments}
This work was supported by The Engineering and Physical Sciences Council, The Medical Research Council, The Alan Turing Institute and the Alan Turing Institute --- Lloyd's Register Foundation Programme on Data-centric Engineering, under grant numbers EP/L016710/1, EP/R034710/1, EP/T004134/1, EP/N510129/1 and EP/R044732/1 .

\appendix
\section{Comparison of time scales for multinomial resampling and\\ stochastic rounding}\label{app:timescales}

Fix $N\geq1$. Suppose we are given a fixed vector of weights $w_t^{(1:N)}$. Let $\E^{M}$ and $\E^{SR}$ denote expectations with respect to the laws of resampling steps using multinomial resampling and stochastic rounding respectively.
Let $\nu_t^{(1:N)}$ be the resulting offspring counts. 

\begin{proposition}
For any $i \in \{1, \dots, N\}$,
$
\E^{M}[ c_N(t) \mid w_t^{(1:N)} ] \geq \E^{SR}[ c_N(t) \mid w_t^{(1:N)} ].
$
\end{proposition}
\begin{proof}
It is sufficient to show that $\E^{M}[ (\nu_t^{(i)})_2 \mid w_t^{(1:N)} ] \geq \E^{SR}[ (\nu_t^{(i)})_2 \mid w_t^{(1:N)} ]$ for each $i \in \{1, \dots, N\}$. 
Using properties of the Multinomial distribution \cite{mosimann1962}, we have
\begin{equation*}
\E^{M}[ (\nu_t^{(i)})_2 \mid w_t^{(1:N)} ]  = N(N-1) (w_t^{(i)})^2 .
\end{equation*}
Directly from Definition \ref{defn:stochround}, we calculate the corresponding quantity in the case of stochastic rounding to be
\begin{align*}
\E^{SR}[ (\nu_t^{(i)})_2 \mid w_t^{(1:N)} ] 
&= \flnw (\flnw -1) (1 - Nw_t^{(i)} + \flnw) \\
&\qquad\qquad + (\flnw +1) \flnw (Nw_t^{(i)} - \flnw) \\
&= \flnw \left( 2(Nw_t^{(i)} - \flnw) + \flnw -1 \right) .
\end{align*}
We define the difference $\Delta_i := \E^{M}[ (\nu_t^{(i)})_2 \mid w_t^{(i)} ] - \E^{SR}[ (\nu_t^{(i)})_2 \mid w_t^{(i)} ] $,
and show that $\Delta_i \geq 0$ for all $0 \leq w_t^{(i)} \leq 1$. Partition the interval $[0,1]$ into the half-open intervals $[k/N, (k+1)/N)$ for each $k\in \{0, \dots, N-1\}$, plus the singleton $\{1\}$.

If $w_t^{(i)}=1$, it follows easily that $\E^{M}[ (\nu_t^{(i)})_2 \mid w_t^{(1:N)} ] = \E^{SR}[ (\nu_t^{(i)})_2 \mid w_t^{(1:N)} ] = N(N-1) $. 
For the other cases, suppose $k/N \leq w_t^{(i)} < (k+1)/N $ for some $k \in \{0, \dots, N-1\}$. Then $\flnw = k$, and
\begin{align*}
\Delta_i &= (Nw_t^{(i)} - k)^2 - N (w_t^{(i)})^2 + k\\
&= N(N-1) \left\{ \left( w_t^{(i)} -\frac{k}{N-1} \right)^2 - \frac{k^2}{(N-1)^2} + \frac{k(k+1)}{N(N-1)} \right\}\\
&= N(N-1) \left( w_t^{(i)} -\frac{k}{N-1} \right)^2 + k\left(1-\frac{k}{N-1}\right) \\
&\geq k\left(1-\frac{k}{N-1}\right)
\geq 0 .
\end{align*}
For each $N\geq 2$, any $w_t^{(i)} \in [0,1]$ falls into one of the above cases, so for any fixed vector $w_t^{(1:N)}$ of weights, we have that $\Delta_i \geq 0$ for all $i$. For $N=1$ the result is immediate. This concludes the proof.
\end{proof}

\section{Proof of finite time-scale condition for corollaries}\label{app:finiteness}

\begin{lemma}\label{thm:SR_nontriviality}
Consider an SMC algorithm using any stochastic rounding as its resampling scheme.
Suppose that $\varepsilon \leq q_t(x, x^\prime) \leq \varepsilon^{-1}$ uniformly for some $\varepsilon \in (0,1]$, and that there exist $\zeta >0$ and $\delta \in (0,1)$ such that $\Prob[ \max_i w_t^{(i)} - \min_i w_t^{(i)} \geq 2\delta/N \mid \mathcal{F}_{t-1} ] \geq \zeta$ for infinitely many $t$. Then, for all $N>1$, $\Prob[ \tau_N(t) = \infty ] =0$ for all finite $t$.
\end{lemma}

\begin{proof}
Let $\mathcal{H}_t$ be defined as in \eqref{eq:defn_Ht}. The first step is to show that whenever $\max_i w_t^{(i)} \geq (1+\delta)/N$, $\Prob[  c_N(t) > 2/N^2 | \mathcal{H}_t ] = \Prob[ c_N(t) \neq 0 | \mathcal{H}_t ]$ is bounded below uniformly in $t$.
For this purpose we need consider only weight vectors such that $w_t^{(i)} \in (0,2/N)$ for all $i$; otherwise $\Prob[ c_N(t) \neq 0 | \mathcal{H}_t ] =1$ by the definition of stochastic rounding.

Denote $\mathcal{S}_{N-1}^\delta = \{ w^{(1:N)} \in \mathcal{S}_{N-1} :  \forall i, \, 0 <w^{(i)} <2/N ;\, \max_i w^{(i)} \geq (1 + \delta)/N \}$ for any $\delta \in (0, 1)$, where $\mathcal{S}_{k}$ denotes the $k$-dimensional probability simplex.
Fix arbitrary $w_t^{(1:N)} \in \mathcal{S}_{N-1}^\delta$. Set $i^\star = \arg\max_i w_t^{(i)}$ and denote $\mathcal{I} = \{i \in \{1,\dots,N\} : w^{(i)} > 1/N \}$.
Since all weights are in $(0, 2/N)$, for $i \in \mathcal{I}, \nu_t^{(i)} \in \{1,2\}$ and for $i \notin \mathcal{I}, \nu_t^{(i)} \in \{0,1\}$; and since the offspring counts must sum to $N$, we can write
\begin{align}\label{eq:smallcN_istar}
\Prob[ c_N(t) \leq 2/N^2 | \mathcal{H}_t ]
&= \Prob[ \nu_t^{(i)} =1 \,\forall i\in\{1,\dots,N\} | \mathcal{H}_t ] \notag\\
&= \Prob[ \nu_t^{(i)} =1 \,\forall i\in \mathcal{I} | \mathcal{H}_t ] \notag\\
&= \prod_{i \in \mathcal{I}} \Prob[ \nu_t^{(i)} =1 | \nu_t^{(j)}=1 \,\forall j \in \mathcal{I}: j<i; \mathcal{H}_t ] \notag\\
&= \Prob[ \nu_t^{(i^\star)} =1 | \mathcal{H}_t ] \prod_{\substack{i \in \mathcal{I} \\ i \neq i^\star}} \Prob[ \nu_t^{(i)} =1 | \nu_t^{(i^\star)}=1; \nu_t^{(j)}=1 \,\forall j \in \mathcal{I}: j<i ; \mathcal{H}_t ] \notag\\
&\leq \Prob[ \nu_t^{(i^\star)} =1 | \mathcal{H}_t ] .
\end{align}
The final inequality holds with equality when $|\mathcal{I}| =1$, i.e.\ the only weight larger than $1/N$ is $w_t^{(i^\star)}$.
Thus $\Prob[ c_N(t) > 2/N^2 | \mathcal{H}_t ]$ is minimised on $\mathcal{S}_{N-1}^\delta$ when only one weight is larger than $1/N$, in which case the values of the other weights do not affect this probability. 

Define $w_{\delta^\prime} = \{(1,\dots,1) + \delta^\prime e_{i^\star} - \delta^\prime e_{j^\star} \} /N$ for fixed $i^\star \neq j^\star$ and $\delta^\prime \in (0,1)$, where $e_i$ denotes the $i$th canonical basis vector in $\mathbb{R}^N$. 
As in the proof of Corollary~\ref{thm:stochrounding}, define $p_0^{(i)} = \Prob[ \nu_t^{(i)} = \flnw \mid \mathcal{H}_t ]$ and $p_1^{(i)} = \Prob[ \nu_t^{(i)} = \flnw +1 \mid \mathcal{H}_t ]$. Then from \eqref{eq:smallcN_istar} we have
\begin{equation*}
\Prob[ c_N(t) > 2/N^2 \mid \mathcal{H}_t, w_t^{(1:N)} = w_{\delta^\prime} ]
= 1- \Prob[ \nu_t^{(i^\star)} = 1 \mid \mathcal{H}_t, w_t^{(1:N)} = w_{\delta^\prime} ]
= p_1^{(i^\star)},
\end{equation*}
evaluated on $w_{\delta^\prime}$.
We will need a lower bound on $p_1^{(i^\star)}$ when $w_t^{(1:N)} = w_{\delta^\prime}$. 
We first derive expressions for $p_0^{(i)}$ and $p_1^{(i)}$ up to a constant, then use $p_0^{(i)} + p_1^{(i)} =1$ to get a normalised bound. We have
\begin{align*} 
p_0^{(i)} &= C (1- N w_t^{(i)} + \flnw) \\
&\qquad \times \sum_{\substack{a_{1:N} \in \{1,\dots,N\}^N : \\ |\{j: a_j=i\}|=\flnw }}
\Prob\left[ a_t^{(1:N)} = a_{1:N} \mid \nu_t^{(i)}, w_t^{(1:N)} \right]
\prod_{k=1}^N q_{t-1}( X_t^{(a_k)}, X_{t-1}^{(k)} ) ,\\
p_1^{(i)} &= C (N w_t^{(i)} - \flnw) \\
&\qquad \times \sum_{\substack{a_{1:N} \in \{1,\dots,N\}^N : \\ |\{j: a_j=i\}|=\flnw +1 }}
\Prob\left[ a_t^{(1:N)} = a_{1:N} \mid \nu_t^{(i)}, w_t^{(1:N)} \right]
\prod_{k=1}^N q_{t-1}( X_t^{(a_k)}, X_{t-1}^{(k)} ) .
\end{align*}
Applying the bounds on $q_t$, we have
\begin{align*}
C (1- N w_t^{(i)} + \flnw) \varepsilon^N &\leq p_0^{(i)} \leq C (1- N w_t^{(i)} + \flnw) \varepsilon^{-N} ,\\
C (N w_t^{(i)} - \flnw) \varepsilon^N &\leq p_1^{(i)} \leq C (N w_t^{(i)} - \flnw) \varepsilon^{-N}
\end{align*}
from which we construct the normalised bound
\begin{equation*}
p_1^{(i)} \geq \frac{ (Nw_t^{(i)} - \flnw) \varepsilon^{N} }{ (Nw_t^{(i)} - \flnw) \varepsilon^{-N} + (1- Nw_t^{(i)} +\flnw) \varepsilon^{-N}}
= (Nw_t^{(i)} - \flnw) \varepsilon^{2N} .
\end{equation*}
When $w_t^{(1:N)} = w_{\delta^\prime}$, we have $w_t^{(i^\star)} = (1+\delta^\prime)/N$, so $p_1^{(i^\star)} \geq \delta^\prime \varepsilon^{2N}$,
which is increasing in $\delta^\prime$.
We conclude that $\Prob[ c_N(t) > 2/N^2 | \mathcal{H}_t, \max_i w_t^{(i)} \geq (1+\delta)/N ] \geq \min_{\delta^\prime \geq \delta} \delta^\prime \varepsilon^{2N} = \delta \varepsilon^{2N}$.

A slight modification of this argument yields $\Prob[ c_N(t) > 2/N^2 | \mathcal{H}_t, \min_i w_t^{(i)} \leq (1-\delta)/N ] \geq \delta \varepsilon^{2N} $.
Whenever $\max_i w_t^{(i)} - \min_i w_t^{(i)} \geq 2\delta/N$, either $\max_i w_t^{(i)} \geq (1+\delta)/N$ or $\min_i w_t^{(i)} \leq (1-\delta)/N$, so we have 
$\Prob[ c_N(t) > 2/N^2 | \mathcal{H}_t, \max_i w_t^{(i)} - \min_i w_t^{(i)} \geq 2\delta/N ] \geq \delta \varepsilon^{2N}$.
Thus 
\begin{equation*}
\Prob[ c_N(t)>2/N^2 \mid \mathcal{H}_t ] \geq \delta \varepsilon^{2N}\1{\max_i w_t^{(i)} - \min_i w_t^{(i)} \geq 2\delta/N} .
\end{equation*}
Using the D-separation established in Appendix \ref{app:dseparation} combined with the tower property, we have
\begin{align*}
\Prob[ c_N(t)>2/N^2 \mid \mathcal{F}_{t-1} ]
&=\Et\left[ \Prob[ c_N(t)>2/N^2 \mid \mathcal{H}_t, \mathcal{F}_{t-1} ] \right]
=\Et\left[ \Prob[ c_N(t)>2/N^2 \mid \mathcal{H}_t ] \right] \\
&\geq \delta \varepsilon^{2N} \Prob[ \max_i w_t^{(i)} - \min_i w_t^{(i)} \geq 2\delta/N \mid \mathcal{F}_{t-1} ] ,
\end{align*}
which is bounded below by $ \zeta \delta \varepsilon^{2N} $ for infinitely many $t$. 
Hence,
\begin{equation*}
\sum_{t=0}^\infty \Prob[ c_N(t) > 2/N^2 \mid \mathcal{F}_{t-1} ] = \infty .
\end{equation*}
By a filtered version of the second Borel--Cantelli lemma (see for example \cite[Theorem 4.3.4]{durrett2019}), this implies that $c_N(t) >2/N^2$ for infinitely many $t$, almost surely.
This ensures, for all $t <\infty$, that $\Prob\left[ \exists s<\infty : \sum_{r=1}^s c_N(r) \geq t \right] =1$, which by definition of $\tau_N(t)$ is equivalent to $\Prob[ \tau_N(t) = \infty ] =0$.
\end{proof}

\begin{lemma}\label{thm:CSMC_nontriviality}
Consider a conditional SMC algorithm using multinomial resampling, satisfying the standing assumption and \eqref{eq:gq_bounds_csmc}. 
Then, for all $N>2$, $\Prob[ \tau_N(t) = \infty ]=0$ for all finite $t$.
\end{lemma}

\begin{proof}
Since $c_N(t) \in [0,1]$ almost surely and has strictly positive expectation, for any fixed $N$ the distribution of $c_N(t)$ with given expectation that maximises $\Prob[ c_N(t)=0 \mid \mathcal{F}_{t-1} ]$ is two atoms, at 0 and 1 respectively. To ensure the correct expectation, the atom at 1 should have mass $\Prob[ c_N(t)=1 \mid \mathcal{F}_{t-1} ] = \Et [ c_N(t) ]$, which is bounded below by \eqref{eq:CSMC_cN_LB}.
If $c_N(t) > 0$ then $c_N(t) \geq 2/(N)_2 > 2/N^2$. Hence, in general $\Prob[ c_N(t) > 2/N^2 \mid \mathcal{F}_{t-1} ] \geq \Et [c_N(t)]$. Applying \eqref{eq:CSMC_cN_LB}, we have for any finite $N$,
\begin{equation*}
\sum_{t=0}^\infty \Prob[ c_N(t) > 2/N^2 \mid \mathcal{F}_{t-1} ]
\geq \sum_{t=0}^\infty \Et [ c_N(t) ]
\geq \sum_{t=0}^\infty \frac{1}{N^2} \left\{\frac{2\varepsilon^2}{a^2} + \frac{(N-2)\varepsilon^4}{a^4}  \right\}
= \infty
\end{equation*}
By an argument analogous to the conclusion of Lemma \ref{thm:SR_nontriviality}, $\Prob[ \tau_N(t) = \infty ] =0$ for all $t < \infty$.
\end{proof}

\section{D-separation argument to establish conditional independence of $a_t^{(1:N)}$ and $\mathcal{F}_{t-1}$ given $\mathcal{H}_t$}\label{app:dseparation}
Figure \ref{fig:cond_indep_graph} shows part of the conditional dependence graph implied by Algorithm \ref{alg:SMC}. Our aim is to find a $\sigma$-algebra $\mathcal{H}_t$ at each time $t$ that separates the ancestral process (encoded by $a_t^{(1:N)}$) from the filtration $\mathcal{F}_{t-1}$. That is, $a_t^{(1:N)}$ is conditionally independent of $\mathcal{F}_{t-1}$ given $\mathcal{H}_t$.
By a D-separation argument (see \cite{Verma1988}), the nodes highlighted in grey suffice as the generator of $\mathcal{H}_t$. That is, for each $t$, we take
\begin{equation*}
\mathcal{H}_t = \sigma(X_{t-1}^{(1:N)}, X_t^{(1:N)}, w_{t-1}^{(1:N)}, w_t^{(1:N)} ).
\end{equation*}
Notice that $\nu_t^{(1:N)}$ can be expressed as a non-injective function of $a_t^{(1:N)}$, and as such carries less information.
\begin{figure}[ht]
\centering
\begin{tikzpicture}[>=stealth]
\filldraw[gray!20, rounded corners] (3.2,0.5)--(8.6,0.5)--(8.6,-2.5)--(3.2,-2.5)--cycle;
\node[gray!70] at (8.9,0.3) {$\mathcal{H}_t$};
\node at (-2,0) {...};
\node at (-2,-2) {...};
\node at (-2,-4) {...};
\node at (-2,-6) {...};
\node at (0,0) {$X_{t+1}^{(1:N)}$};
\node at (0,-2) {$w_{t+1}^{(1:N)}$};
\node at (0,-4) {$a_{t+1}^{(1:N)}$};
\node at (0,-6) {$\nu_{t+1}^{(1:N)}$};
\node at (4,0) {$X_{t}^{(1:N)}$};
\node at (4,-2) {$w_{t}^{(1:N)}$};
\node at (4,-4) {$a_{t}^{(1:N)}$};
\node at (4,-6) {$\nu_{t}^{(1:N)}$};
\node at (8,0) {$X_{t-1}^{(1:N)}$};
\node at (8,-2) {$w_{t-1}^{(1:N)}$};
\node at (8,-4) {$a_{t-1}^{(1:N)}$};
\node at (8,-6) {$\nu_{t-1}^{(1:N)}$};
\node at (10,0) {...};
\node at (10,-2) {...};
\node at (10,-4) {...};
\node at (10,-6) {...};
\draw [rounded corners, dashed, gray] (11,-6.6)--(7,-6.6)--(7,-5.5)--(11,-5.5);
\draw [rounded corners, dashed, gray] (11,-6.8)--(3,-6.8)--(3,-5.3)--(11,-5.3);
\draw [rounded corners, dashed, gray] (11,-7)--(-1,-7)--(-1,-5.1)--(11,-5.1);
\node[gray] at (7.5,-6.4) {\footnotesize{$\mathcal{F}_{t-1}$}};
\node[gray] at (3.3,-6.6) {\footnotesize{$\mathcal{F}_{t}$}};
\node[gray] at (-0.5,-6.8) {\footnotesize{$\mathcal{F}_{t+1}$}};
\draw[->] (0.5,0)--(3.4,0);
\draw[->] (0.5,0)--(3.4,-2);
\draw[->] (0.5,-4)--(3.4,-2.1);
\draw[->] (0.5,-4)--(3.4,-0.1);
\draw[->] (4.5,0)--(7.4,0);
\draw[->] (4.5,0)--(7.4,-2);
\draw[->] (4.5,-4)--(7.4,-2.1);
\draw[->] (4.5,-4)--(7.4,-0.1);
\draw[->] (0,-0.3)--(0,-1.7);
\draw[->] (0,-2.3)--(0,-3.7);
\draw[->] (0,-4.3)--(0,-5.7);
\draw[->] (4,-0.3)--(4,-1.7);
\draw[->] (4,-2.3)--(4,-3.7);
\draw[->] (4,-4.3)--(4,-5.7);
\draw[->] (8,-0.3)--(8,-1.7);
\draw[->] (8,-2.3)--(8,-3.7);
\draw[->] (8,-4.3)--(8,-5.7);
\end{tikzpicture}
\caption{Part of the conditional dependence graph implied by Algorithm \ref{alg:SMC}. The direction of time is from left to right. The reverse-time filtration is indicated by the dashed areas. The nodes highlighted in grey generate the separatrix $\mathcal{H}_t$ between $a_t^{(1:N)}$ and $\mathcal{F}_{t-1}$.}
\label{fig:cond_indep_graph}
\end{figure}
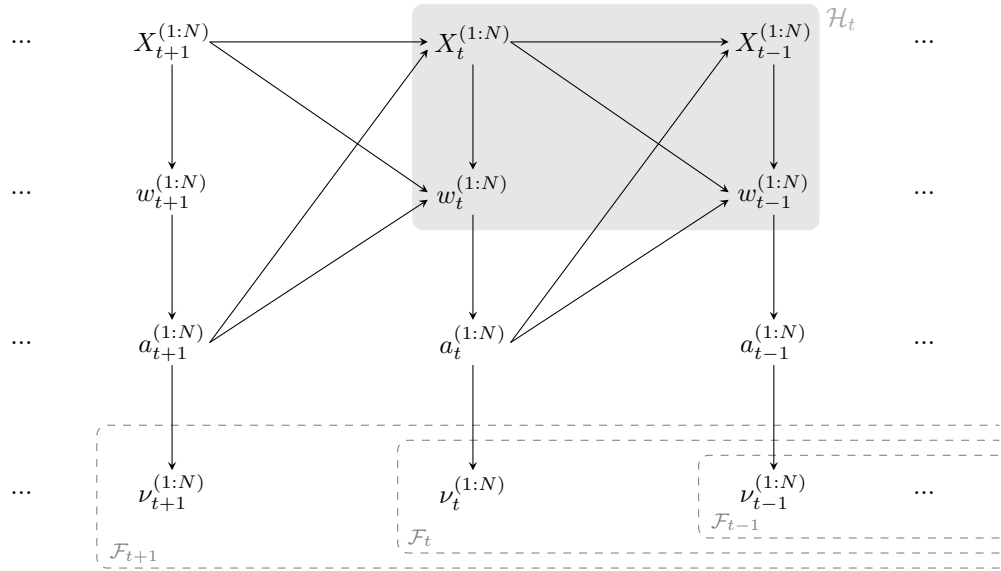

\bibliography{../smc.bib}
\end{document}